\documentclass[11pt]{article}
\usepackage[margin=1in]{geometry}

\usepackage{amsmath,amsfonts,amssymb,amsthm}
\usepackage{mathtools}
\usepackage{scrlfile}
\usepackage{thmtools,thm-restate}
\makeatletter
\newcommand*\rel@kern[1]{\kern#1\dimexpr\macc@kerna}
\newcommand*\widebar[1]{%
  \begingroup
  \def\mathaccent##1##2{%
    \rel@kern{0.8}%
    \overline{\rel@kern{-0.8}\macc@nucleus\rel@kern{0.2}}%
    \rel@kern{-0.2}%
  }%
  \macc@depth\@ne
  \let\math@bgroup\@empty \let\math@egroup\macc@set@skewchar
  \mathsurround\z@ \frozen@everymath{\mathgroup\macc@group\relax}%
  \macc@set@skewchar\relax
  \let\mathaccentV\macc@nested@a
  \macc@nested@a\relax111{#1}%
  \endgroup
}
\makeatother

\newtheorem{theorem}{Theorem}

\newtheorem{lemma}[theorem]{Lemma}

\newtheorem{corollary}[theorem]{Corollary}
\newtheorem{definition}[theorem]{Definition}

\theoremstyle{definition}
\newtheorem{open}{Open Problem}

\newcommand{\be}{\begin{equation}}
\newcommand{\ee}{\end{equation}}

\newcommand{\B}{\{0,1\}}

\newcommand{\tOmega}{\tilde{\Omega}}

\DeclareMathOperator{\R}{R}
\DeclareMathOperator{\Q}{Q}

\DeclareMathOperator{\Dom}{Dom}

\DeclareMathOperator{\poly}{poly}

\DeclareMathOperator{\bR}{\mathbb{R}}

\DeclareMathOperator{\cost}{cost}
\DeclareMathOperator{\err}{err}

\newcommand{\triv}{\textsc{Triv}}

\newcommand{\bN}{\mathbb{N}}
\newcommand{\bE}{\mathbb{E}}

\usepackage[hidelinks]{hyperref}
\hypersetup{
    pdftitle={}, 
    pdfauthor={}, 
    colorlinks=true, 
    linkcolor=blue, 
    citecolor=blue, 
    urlcolor=blue 
}

\newcommand{\eq}[1]{\hyperref[eq:#1]{(\ref*{eq:#1})}}
\renewcommand{\sec}[1]{\hyperref[sec:#1]{Section~\ref*{sec:#1}}}
\newcommand{\thm}[1]{\hyperref[thm:#1]{Theorem~\ref*{thm:#1}}}
\newcommand{\lem}[1]{\hyperref[lem:#1]{Lemma~\ref*{lem:#1}}}
\newcommand{\defn}[1]{\hyperref[def:#1]{Definition~\ref*{def:#1}}}
\newcommand{\prop}[1]{\hyperref[prop:#1]{Proposition~\ref*{prop:#1}}}
\newcommand{\cor}[1]{\hyperref[cor:#1]{Corollary~\ref*{cor:#1}}}
\newcommand{\fig}[1]{\hyperref[fig:#1]{Figure~\ref*{fig:#1}}}
\newcommand{\tab}[1]{\hyperref[tab:#1]{Table~\ref*{tab:#1}}}
\newcommand{\alg}[1]{\hyperref[alg:#1]{Algorithm~\ref*{alg:#1}}}
\newcommand{\app}[1]{\hyperref[app:#1]{Appendix~\ref*{app:#1}}}
\newcommand{\conj}[1]{\hyperref[conj:#1]{Conjecture~\ref*{conj:#1}}}
\newcommand{\chap}[1]{\hyperref[chap:#1]{Chapter~\ref*{chap:#1}}}
\newcommand{\clm}[1]{\hyperref[clm:#1]{Claim~\ref*{clm:#1}}}
\newcommand{\fct}[1]{\hyperref[fct:#1]{Fact~\ref*{fct:#1}}}
\newcommand{\opn}[1]{\hyperref[opn:#1]{Open Problem~\ref*{opn:#1}}}

\begin{document}

\title{How symmetric is too symmetric for large quantum speedups?}

\author{
Shalev Ben{-}David\\
\small University of Waterloo\\
\small \texttt{shalev.b@uwaterloo.ca}
\and
Supartha Podder\\
\small University of Ottawa\\
\small \texttt{spodder@uottawa.ca}
}

\date{}
\maketitle

\begin{abstract}
Suppose a Boolean function $f$ is symmetric under a group action
$G$ acting on the $n$ bits of the input. For which $G$ does this mean
$f$ does not have an exponential quantum speedup? Is there a characterization
of how rich $G$ must be before the function $f$ cannot have enough structure
for quantum algorithms to exploit?

In this work, we make several steps towards understanding the group actions $G$
which are ``quantum intolerant'' in this way.
We show that sufficiently transitive
group actions do not allow a quantum speedup, and that a ``well-shuffling''
property of group actions -- which happens to be preserved by several natural
transformations -- implies a lack of super-polynomial speedups for functions
symmetric under the group action. Our techniques are motivated by
a recent paper by Chailloux (2018), which deals with the case where $G=S_n$.

Our main application is for graph symmetries: we show that any Boolean
function $f$ defined on the adjacency matrix of a graph (and symmetric under
relabeling the vertices of the graph) has a power $6$ relationship between
its randomized and quantum query complexities, even if $f$ is a partial function.
In particular, this means no graph property testing problems can have super-polynomial
quantum speedups, settling an open problem of Ambainis, Childs, and Liu (2011).
\end{abstract}


\section{Introduction}

One of the most fundamental questions in the field of quantum
computing is the question of when quantum algorithms
substantially outperform classical ones.
While polynomial quantum speedups are known in
many settings, super-polynomial quantum speedups are known
(or even merely conjectured) for only a few select problems.
An important lesson in the field has been that exponential
quantum speedups only occur for certain ``structured'' problems:
problems such as period-finding (used in Shor's factoring algorithm
\cite{Sho97}), or Simon's problem \cite{Sim97}, in which the input
is known in advance to have a highly restricted form. In contrast,
for ``unstructured'' problems such as blackbox search or
$\mathsf{NP}$-complete problems,
only polynomial speedups are known (and in some models,
it can be formally shown that only polynomial speedups are possible).

In this work, we are interested in formalizing and characterizing
the structure necessary for fast quantum algorithms; in particular,
we study the \emph{types of symmetries} a Boolean function can have
while still exhibiting super-polynomial quantum speedups.

\subsection{Previous work}

Despite the strong intuition in the field that structure is necessary
for exponential quantum speedups, only a handful of works have attempted
to formalize this intuition and characterize the necessary structure.
All of them study the problem in the query complexity (black-box) model
of quantum computation, which is a natural framework in which
both period-finding and Simon's problem can be formally shown to
give exponential quantum speedups (see \cite{BdW02} for a survey
of query complexity,
or \cite{Cle04} for a formalization of period-finding specifically).

In the query complexity model, the goal is to compute
a Boolean function $f:\Sigma^n\to\B$ using as few queries to
the bits of the input $x\in\Sigma^n$ as possible. Here $\Sigma$ is some
finite alphabet, and each query specifies an index $i\in[n]$
and receives the response $x_i\in\Sigma$. A query algorithm,
which may depend on $f$ but not on $x$,
must output $f(x)$ (to worst-case bounded error) after as few
queries as possible. Quantum query algorithms are allowed
to make queries in superposition; we are interested in
how much advantage this gives them over randomized classical
algorithms (for a formal definition of these notions, see \cite{BdW02}).

Beals, Buhrman, Cleve, Mosca, and de Wolf \cite{BBC+01} showed
that all \emph{total Boolean functions} $f:\Sigma^n\to\B$
have a polynomial relationship
between their classical and quantum query complexities
(which we denote $\R(f)$ and $\Q(f)$ respectively).
This means that super-polynomial speedups are not possible in
query complexity unless we impose a promise on the input: that is,
unless we define $f:P\to\B$ with $P\subseteq\Sigma^n$, and allow
an algorithm computing $f$ to behave arbitrarily on inputs
outside of the promise set $P$. For such promise problems
(also called partial functions), provable exponential quantum speedups
are known; this is the setting in which Simon's problem and
period-finding reside.

The question, then, is whether we can say anything about the structure
necessary for a partial Boolean function $f$ to exhibit
a super-polynomial quantum speedup. Towards this end,
Aaronson and Ambainis \cite{AA14} showed
that \emph{symmetric} functions do not allow super-polynomial quantum
speedups, even with a promise. Chailloux \cite{Cha18} improved
this result by improving the degree of the polynomial
relationship between randomized and quantum algorithms for symmetric
functions, and removing a technical requirement on the symmetry
of those functions.\footnote{
Aaronson and Ambainis required the function
to be symmetric both under permuting the $n$ bits of the input,
and under permuting the alphabet symbols in $\Sigma$;
Chailloux showed that the latter is not necessary.}

Other work attempted to characterize the structure necessary for
quantum speedups in other ways. \cite{Ben16} Showed that certain
types of symmetric promises do not admit any function with a
super-polynomial quantum speedup, a generalization of \cite{BBC+01}
(who showed this when the promise set is $\Sigma^n$).
\cite{AB16} Showed that small promise sets, which contain
only $\poly(n)$ inputs out of $|\Sigma|^n$, also do not admit
functions which separate quantum and classical algorithms by
more than a polynomial factor.

\subsection{Our contributions}

In this work, we extend the results of Aaronson-Ambainis and
Chailloux to other symmetry groups. To state our results, we introduce
the following definition.

\begin{definition}
Let $f:P\to\B$ be a function with $P\subseteq\Sigma^n$,
where $\Sigma$ is a finite alphabet and $n\in\bN$. We say
that $f$ is \emph{symmetric} with respect to a group action $G$
acting on domain $[n]$ if for all $x\in P$ and all $\pi\in G$,
the string $x\circ\pi$ defined by $(x\circ\pi)_i\coloneqq x_{\pi(i)}$
satisfies $x\circ\pi\in P$ and $f(x\circ\pi)=f(x)$.
\end{definition}

This definition allows us to talk about more general symmetries
of a Boolean function. The case where $G=S_n$ is the fully-symmetric
group action is the one handled in Chailloux's work \cite{Cha18}:
he showed that $\R(f)=O(\Q(f)^3)$ if $f$ is symmetric under $S_n$.
Aaronson and Ambainis \cite{AA14}
required an even stronger symmetry property.
We note that when $\Sigma$ is large, say $|\Sigma|=n$ or larger,
the class of functions symmetric under $S_n$ is already highly
nontrivial: among others,
it includes functions such as $\textsc{Collision}$,
an important function whose quantum query complexity
was established in \cite{AS04}; $k$-$\textsc{Sum}$,
whose quantum query complexity required the negative-weight
adversary to establish \cite{BS13}; and
$k$-$\textsc{Distinctness}$, whose quantum query complexity is
still open \cite{BKT18}. Additionally, computational
geometry functions such as $\textsc{ClosestPair}$
(a function studied in recent work by Aaronson,
Chia, Lin, Wang, and Zhang \cite{ACL+19}) are typically symmetric
under $S_n$ as well, as the points are usually represented as alphabet
symbols. If we round the alphabet to be finite, the $G=S_n$
case already shows that no computational geometry problem
of this form can have super-cubic quantum speedups.

In this work, we examine what happens when we
relax the full symmetry $S_n$ to smaller symmetry groups $G$.
We introduce some tools for showing that particular classes
of group actions $G$ do not allow super-polynomial quantum speedups;
that is, we provide tools for showing that every $f$ symmetric
with respect to $G$ satisfies $\Q(f)=\R(f)^{\Omega(1)}$.
Our primary application is the following theorem,
in which $G$ is the \emph{graph symmetry}:
the group action acting on strings of length $\binom{k}{2}$
(which represent the possible edges of a graph), which includes
all permutations of the edges which are induced by one of the $k!$
relabelings of the $k$ vertices. Functions which take in
the adjacency matrix of a graph as input, and whose output
depends only on the graph (and not on the labeling of its vertices),
are always symmetric with respect to the graph symmetry $G$.

\begin{theorem}[Informal]\label{thm:main_graph}
Any Boolean function $f$ defined on the adjacency matrix of
a graph (and symmetric with respect to renaming the vertices of the
graph) has a polynomial (power $6$) relationship between
$\R(f)$ and $\Q(f)$. This holds even if $f$ is a partial function.
\end{theorem}

(For a formal version of this theorem, see \cor{formal_graph}.)

This theorem holds even when the alphabet of $f$ is non-Boolean.
We note that this is a strict generalization of the result
of Aaronson and Ambainis \cite{AA14}, since any fully-symmetric
function will necessarily be symmetric under the graph
symmetry as well. It is also a generalization of Chailloux \cite{Cha18},
except that our polynomial degree (power $6$) is larger than
the power $3$ of Chailloux.

This theorem also settles an open problem of Ambainis, Childs, and Liu
\cite{ACL11} at least for the adjacency matrix version of
graph property testing. They asked whether there is any
graph property testing problem with an exponential quantum speedup
over the best possible classical algorithm; our theorem
implies the answer is no.\footnote{An alternative version of
graph property testing is called the adjacency list model,
in which it is possible to directly query the list of neighbors
of each vertex. Query functions on adjacency list
graphs can achieve large quantum speedups -- for example,
the glued trees problem \cite{CCD+03} -- but it is still
open whether \emph{property testing} graph problems
in the adjacency list model can achieve exponential quantum speedups.
We conjecture that such exponential speedups for graph property
testing \emph{do} exist -- perhaps by a modification of
the glued trees problem.}
Indeed, any graph property testing problem is always symmetric
under the graph symmetry group action $G$, which means that
all graph property testing problems satisfy a power $6$ relationship
between their quantum and classical query complexities.

Our tools apply to other group actions as well. We show that
highly transitive group actions also are not consistent with
exponential quantum speedups.

\begin{theorem}[informal]\label{thm:main_transitive}
Let $G$ be a $n^{\Omega(1)}$-transitive group action on $[n]$,
and let $f$ be a (possibly partial) Boolean function
on strings of length $n$ which is symmetric
under $G$. Then $\Q(f)=\R(f)^{\Omega(1)}$.
\end{theorem}

(A formal version of this theorem can be found in \cor{formal_transitive}.)

We note that this theorem does not subsume the previous one,
as graph symmetry group actions are not even $2$-transitive.

Further, we are able to generalize our results to directed graph
symmetries, hypergraph symmetries, and bipartite graph symmetries. We also provide a reasonably clean framework
in which we prove these results, and show that various natural
operations on group actions preserve the
lack-of-exponential-quantum-speedup property.

Finally, we examine the other direction, and exhibit some classes
of group actions whose symmetries do allow exponential quantum speedups.
In particular, we show the following.

\begin{theorem}[informal]
\label{thm:large_order_no_speedup}
The order of the group action does not characterize whether it allows
exponential quantum speedups. In particular, there is a group
action $G_n$ on $[n]$ and a different group action $H_n$ on $[n]$
such that $|G_n|=|H_n|=n^{\Omega(n)}$, but every
function $f$ that's symmetric under $G$ has $\R(f)=O(\Q(f)^3)$
while there exists a function $f$ that's symmetric under $H$
and has $\Q(f)=O(1)$ and $\R(f)=n^{\Omega(1)}$.
\end{theorem}
(A formal version of this theorem can be found in
\thm{large_order_no_speedup_formal}.)

This theorem says that even very large group actions may still be
consistent with exponential quantum speedups; to characterize
the group actions which do not allow super-polynomial quantum
advantage, we must use some richness measure other than the order
of the group action (and also other than transitivity,
as some $1$-transitive group actions allow exponential
quantum speedups and some don't). We leave such a characterization
as an intriguing open problem for future work.

\begin{open}\label{opn:oq1}
Is there a clean combinatorial characterization of the classes
of group actions $G$ that allow super-polynomial quantum speedups,
and the ones that don't? For example, is there a combinatorial
measure $M(G)$ such that any function $f$ symmetric under $G$
satisfies something like $\R(f)=\Q(f)^{O(M(G))}$, and also such
that there always exists a function $g$ symmetric under $G$
for which $\R(g)=\Q(g)^{\Omega(M(G))}$?
\end{open}

We view \opn{oq1} as an important direction for understanding
the nature of quantum speedups.
 
\subsection{Our techniques}

Our main tool is a simple observation from \cite{Cha18}.
Suppose that $f$ is symmetric under the full symmetric group action $S_n$.
Zhandry \cite{Zha13} showed that distinguishing
a random permutation from $S_n$ from a random small-range
function $\alpha:[n]\to[n]$ with $|\alpha([n])|=r$ requires
$\Omega(r^{1/3})$ quantum queries\footnote{Actually,
we will show that a version of Zhandry's result that is sufficient
for our purposes follows easily from the collision lower bound,
so his techniques are not necessary for our results.}.
Now, if $Q$ was a quantum
algorithm solving $f$ using $T$ queries, then $Q$ also outputs $f(x)$
on input $x\circ \pi$ (the input $x$ with bits shuffled according to
$\pi$) for any $\pi\in S_n$, since $f$ is symmetric under $S_n$.
In particular, $Q(x\circ\pi)$ for a random $\pi\in S_n$ still
outputs $f(x)$. However, the $T$-query algorithm $Q$ cannot
distinguish a random $\pi\in S_n$ from a random function
$\alpha:[n]\to[n]$ with range $r\approx T^3$; hence $Q$
must output $f(x)$ to constant error even when run on
$x\circ\alpha$ for a random small-range function $\alpha$.
This property can be used to simulate $Q$ classically: a classical
algorithm $R$ will simply sample a small-range function $\alpha$,
explicitly query the entire string $x\circ \alpha$ (possible
to do using $O(T^3)$ queries since $\alpha$ has range only $O(T^3)$),
and then simulate $Q$ on the string $x\circ\alpha$. This is
an $O(T^3)$-query classical algorithm for computing $f$,
created out of a $T$-query quantum algorithm for $f$.

The above trick can be generalized from the fully-symmetric group
action $S_n$ to any other group action $G$, so long as we can show
that it is hard for a $T$-query quantum algorithm to distinguish $G$ from
a small-range function with range $O(\poly(T))$.
The question of whether there exists an arbitrary symmetric function
$f$ with a quantum speedup is therefore reduced to the question
of whether the concrete task of distinguishing $G$ from
the set of all small-range functions can be done quickly
using a quantum algorithm. That is, if $D_{n,r}$ is the set
of all strings in $[n]^n$ which use only $r$ unique symbols,
then we care about the quantum query cost of distinguishing
$D_{n,r}$ from $G$; if this cost is $r^{\Omega(1)}$, then
no function which is symmetric under $G$ can exhibit a super-polynomial
quantum speedup. In this case, we call $G$ \emph{well shuffling}.
We show that the well-shuffling property is preserved under various
operations one might perform on a group action, and that these
operations allow us to prove many group actions are well-shuffling
simply by reduction to $S_n$.

\section{Preliminaries}

\subsection{Query complexity}

We start with introducing some standard notation from query complexity.
A \emph{Boolean function} will be a $\B$-valued function $f$ on
strings of length $n$, with $n\in\bN$. We will use $\Dom(f)$
to denote the domain of $f$, and we will always have
$\Dom(f)\subseteq\Sigma^n$ where $\Sigma$ is a finite alphabet.
The function $f$ is called \emph{total} if $\Dom(f)=\Sigma^n$,
and otherwise it is called \emph{partial}.

For a (possibly partial) Boolean function $f$, we use
$\R_\epsilon(f)$ to denote its \emph{randomized query complexity
to error $\epsilon$},
as defined in \cite{BdW02}. This is the minimum number of queries
required in the worst case by a randomized algorithm which computes
$f$ to worst-case error $\epsilon$. We use $\Q_\epsilon(f)$ to denote the
\emph{quantum query complexity to error $\epsilon$} of $f$,
also defined in \cite{BdW02}.
This is the minimum number of queries required in the worst case
by a randomized algorithm which computes $f$ to worst-case error
$\epsilon$. When $\epsilon=1/3$, we omit it and simply write
$\R(f)$ and $\Q(f)$.

An important tool for lower bounding query complexity
is the minimax theorem, the original version of which was given
by Yao for zero-error (Las Vegas) randomized algorithms \cite{Yao77}.
Here we will need a bounded-error, quantum version
of the minimax theorem. Bounded-error versions of the minimax
theorem can be shown using linear programming duality
(see also \cite{Ver98}
who proved a minimax theorem in the setting where both the error
and the expected query complexity are measured against the same hard
distribution).
A similar technique works for quantum query complexity;
this result is folklore, and we prove it in
\app{minimax}.

\begin{restatable}[Minimax for bounded error quantum algorithms]{lemma}{minimax}
\label{lem:minimax}
Let $f$ be a (possibly partial) Boolean function with
$\Dom(f)\subseteq\Sigma^n$, and let $\epsilon\in(0,1/2)$.
Then there is a distribution $\mu$ supported on $\Dom(f)$
which is \emph{hard} for $f$ in the following sense:
any quantum algorithm using fewer than $\Q_\epsilon(f)$
quantum queries for computing $f$ must have average error
$>\epsilon$ on inputs sampled from $\mu$.
\end{restatable}

Note that achieving average error $\epsilon$ against a known distribution
$\mu$ is always easier than achieving worst-case error $\epsilon$; the minimax
theorem says that there is a hard distribution against which achieving average
error $\epsilon$ is just as hard as achieving worst-case error $\epsilon$.

\subsection{Group actions}

We review some basic definitions about group actions.

\begin{definition}[Group action]
A \emph{group action} is a pair $(D,G)$ where $D$ is a set
and $G$ is a set of bijections $\pi\colon D\to D$, such that
$G$ forms a group under composition (i.e.\ $G$ contains the identity
function and is closed under composition and inverse of the bijections).
We will often denote a group action simply by $G$, with the domain $D$
being implicit.
\end{definition}

In other words, a group action is simply a set of permutations
of a domain $D$ which is closed under composition and inverse.
In this work we will generally take $D=[n]$, where $[n]$ denotes
the set $\{1,2,\dots,n\}$ for $n\in\bN$. The set $[n]$ will
represent the indices of an input string, or equivalently,
the queries an algorithm is allowed to make.

We define orbits and transitivity of group actions, both of which are standard
definitions.

\begin{definition}[Orbit]
Let $G$ be a group action on domain $D$, and let $i\in D$.
Then the \emph{orbit} of $i$ is the set $\{\,\pi(i):\pi\in G\,\}$.
A subset of $D$ is an orbit of $G$ if it is the orbit
of some $i\in D$ with respect to $G$.
\end{definition}

\begin{definition}[Transitivity]
We say that a group action $G$ on domain $D$ is \emph{$k$-transitive} if
for all distinct $i_1,i_2,\dots,i_k\in D$ and distinct
$j_1,j_2,\dots,j_k\in D$, there exists some $\pi\in G$ such that
$\pi(i_t)=j_t$ for all $t=1,2,\dots,k$.
\end{definition}

\subsection{Symmetric functions}

We introduce some notation that will be used throughout this
paper to talk about symmetric functions.

\begin{definition}[Notation for permuting strings]
Let $\pi$ be a permutation on $[n]$, and let $x\in\B^n$.
We use $x\circ \pi$ to denote the string whose characters have been
permuted by $\pi$; that is, $(x\circ\pi)_i\coloneqq x_{\pi(i)}$.
More generally, $x\circ \pi$ is similarly defined when $\pi$
is merely a function $[n]\to[n]$ rather than a permutation.
\end{definition}

Note that if we view a string $x\in\Sigma^n$ as a function
$[n]\to\Sigma$ with $x(i)\coloneqq x_i$, then $x\circ \pi$ is simply
the usual function composition of $x$ and $\pi$.
This notation allows us to easily define symmetric functions.

\begin{definition}[Symmetric function]
Let $G$ be a group action on $[n]$, and let $f$ be a (possibly partial)
Boolean function with $\Dom(f)\subseteq\Sigma^n$.
We say $f$ is \emph{symmetric under} $G$
if for all $x\in\Dom(f)$ and all $\pi\in G$
we have $x\circ\pi\in\Dom(f)$ and $f(x\circ\pi)=f(x)$.
\end{definition}

In order for asymptotic bounds such as $\Q(f)=\R(f)^{\Omega(1)}$
to be well-defined, we actually need to talk about
\emph{classes} of functions rather than individual functions.
To do that, we will need to talk about classes of group actions.
We introduce the following definition, which defines, for a class
of group actions $\mathcal{G}$, the set of all functions symmetric
under some group action in $\mathcal{G}$. We denote this set by $F(\mathcal{G})$.

\begin{definition}[Class of symmetric functions]
Let $\mathcal{G}=\{G_i\}_{i\in I}$ be a (possibly infinite)
set of finite group actions, with $G_i$ acting on $[n_i]$
for each $i\in I$. Here $I$ is an arbitrary index set and $n_i\in\bN$
for all $i\in I$. Then define $F(\mathcal{G})$ to be the
set of all (possibly partial) Boolean functions that are symmetric
under some $G_i$. That is, we have $f\in F(\mathcal{G})$
if and only if
$f\colon\Dom(f)\to\B$ is a function with $\Dom(f)\subseteq[m]^n$
for some $n,m\in\bN$, and $f$ is symmetric under $G_i$ for some
$i\in I$ such that $n_i=n$.
\end{definition}

(In the above definition, $[m]$ represents the alphabet $\Sigma$.)

\section{Well-shuffing group actions}

In this section we first define the notion of a
well-shuffling class of group actions,
which will be a class $\mathcal{G}$ of group actions $G$ that are hard
to distinguish from the set of small-range functions via a quantum query algorithm.
We will then show that a well-shuffling class of group actions does not allow
super-polynomial quantum speedups. This result (\thm{shuffle})
converts the task of showing group actions do not allow quantum speedups
into the task of showing those group actions are well-shuffling,
a much simpler objective.

We start by defining the set of small-range strings $D_{n,r}$.

\begin{definition}[Small-range strings]
For $n,r\in\bN$, let $D_{n,r}$ be the set of all strings $\alpha$ in $[n]^n$
for which the number of unique alphabet symbols in $\alpha$ is at most $r$.
\end{definition}

We identify a string $\alpha\in[n]^n$ with a function $[n]\to[n]$.
Then $D_{n,r}$ is the set of all functions $[n]\to[n]$ with range size
at most $r$. Next, we define $\cost(G,r)$ as the quantum query complexity
of distinguishing $G$ from $D_{n,r}$ (where $G$ is a group action acting
on $[n]$).

\begin{definition}[Cost]
Identify a permutation on $[n]$ with a string in $[n]^n$ in which
each alphabet symbol occurs exactly once. Then a group action $G$
on $[n]$ corresponds to a subset of $[n]^n$. For $r<n$, let
$\cost_\epsilon(G,r)$ be the minimum number of quantum queries
needed to distinguish $G$ from $D_{n,r}$ to worst-case error $\epsilon$;
that is, $\cost_\epsilon(G,r)\coloneqq \Q_\epsilon(f)$,
where $f$ has domain $G\cup D_{n,r}\subseteq [n]^n$ and is defined by
$f(x)=1$ if $x\in G$ and $f(x)=0$ if $x\in D_{n,r}$. When $r\ge n$,
we set $\cost_\epsilon(G,r)\coloneqq\infty$. When $\epsilon=1/3$, we omit it
and write $\cost(G,r)$.
\end{definition}

We note that since $\cost_\epsilon(G,r)$ is defined as the worst-case quantum
query complexity of a Boolean function, it satisfies amplification,
meaning that the precise value of $\epsilon$ does not matter so long as
it is a constant in $(0,1/2)$ and so long as we do not care about
constant factors.

We define a well-shuffling class of group actions as follows.

\begin{definition}[Well-shuffling group actions]
\label{def:shuffling}
Let $\mathcal{G}$ be a collection of group actions.
We say $\mathcal{G}$ is \emph{well-shuffling}
if $\cost(G,r)=r^{\Omega(1)}$ for $G\in\mathcal{G}$ and $r\in\bN$.
More explicitly, we say $\mathcal{G}$ is well-shuffling
\emph{with power $a\in\bN$} if there exists $b\in\bN$ such that
$\cost(G,r)\ge r^{1/a}/b$ for all $G\in \mathcal{G}$ and all $r\in\bN$.
\end{definition}

We note that $\cost(G,r)\ge r^{1/a}/b$ is always satisfied when $r$
is greater than or equal to the domain size of $G$, since in that case
$\cost(G,r)=\infty$. Hence to show well-shuffling we only need to worry about
$r$ smaller than $n$, the domain size of the group action $G$.

The following theorem will play a central role in this work:
it shows that a well-shuffling
collection of group actions does not allow super-polynomial quantum speedups.

\begin{theorem}\label{thm:shuffle}
Let $f\colon\Dom(f)\to\B$ be a partial Boolean function on $n\in\bN$ bits,
with $\Dom(f)\subseteq\Sigma^n$
(where $\Sigma$ is a finite alphabet).
Let $G$ be a group action on $[n]$, and
suppose that $f$ is symmetric under $G$. Then there is a universal
constant $c\in\bN$ such that
\[\R(f)\le \min\{\,r\in\bN:\cost(G,r)\ge c\Q(f)\,\}.\]
Consequently, if 
$\mathcal{G}$ is a well-shuffling collection
of group actions with power $a$, then for all $f\in F(\mathcal{G})$ we have
$\R(f)=O(\Q(f)^a)$.
\end{theorem}

In order to prove this theorem, we will need the following minimax theorem for the
cost measure.

\begin{lemma}[Minimax for cost]\label{lem:cost_minimax}
Let $r,n\in\bN$ satisfy $r<n$, let $\epsilon\in[0,1/2)$,
and let $G$ be a group action on $[n]$. Then there is a distribution
$\mu$ on $D_{n,r}$ that is \emph{hard} in the following sense.
Let $\mu'$ be the uniform distribution on $G\subseteq[n]^n$.
Then any quantum algorithm for distinguishing $G$ from
$D_{n,r}$ which uses fewer than $\cost_\epsilon(G,r)$
queries must either make error $>\epsilon$ on average against $\mu$,
or else make error $>\epsilon$ against $\mu'$
(i.e.\ it fails to distinguish $\mu$ from the uniform distribution on $G$).
\end{lemma}

\begin{proof}
Let $f$ be the function which asks to distinguish $G$ from $D_{n,r}$
in the worst case. Then by the minimax theorem (\lem{minimax}), there is a hard
distribution $\nu$ for $f$, such that any quantum algorithm
using fewer than $\Q_\epsilon(f)=\cost_\epsilon(G,r)$ queries
must make more than $\epsilon$ error against $\nu$.
Let $\nu'$ be the distribution we get by applying a uniformly
random permutation from $G$ to a sample from $\nu$.
Then $\nu'$ is still a hard distribution for $f$. Indeed,
if it were not a hard distribution, there would be some quantum
algorithm $Q$ solving $f$ against $\nu'$ using too few queries;
but in that case, we could design an algorithm $Q'$ for solving
$f$ against $\nu$ simply by taking the input $x$, implicitly
applying a uniformly random $\pi$ from $G$ to permute the bits of $x$
(this can be done without querying $x$, simply by redirecting
all future queries $i\in[n]$ through the permutation $\pi$),
and then running $Q$ on the permuted string.

Now, note that composing
a uniformly random permutation from $G$ with an arbitrary (fixed)
permutation from $G$ gives a uniformly random permutation from $G$.
This means that $\nu'$ is some mixture of the uniform distribution $\mu'$
on $G\subseteq[n]^n$ and another distribution $\mu$ on $D_{n,r}$.
Then any algorithm which succeeds on both $\mu$ and $\mu'$ to
error $\epsilon$ will also succeed on $\nu'$ to error $\epsilon$,
from which the desired result follows.
\end{proof}

Using this lemma, we now prove \thm{shuffle}.

\begin{proof}(Of \thm{shuffle}.)
Let $Q$ be a quantum algorithm for $f$ which uses $\Q(f)$ queries.
Amplify it to $Q'$ by repeating $3$ times and taking the majority
vote; then it uses $3\Q(f)$ queries and makes worst-case error
$7/27$ instead of $1/3$.
Using \lem{cost_minimax}, let $\mu$ be the hard distribution on $D_{n,r}$
which is hard to distinguish from $G$ to error $\epsilon$,
where we pick $r$ later and pick $\epsilon$
to be a constant close to $1/2$. Sample $\alpha$ from $\mu$,
and consider the string $x\circ \alpha$
with $(x\circ\alpha)_i=x_{\alpha(i)}$.

Now, $Q'$ succeeds on $f$ to error $7/27$, and $f$ is invariant
under $G$, so $Q'$ outputs $f(x)$ to
error $7/27$ when run on $x\circ\pi$ for each $\pi\in G$.
In particular, consider picking $\pi$ from $G$ uniformly at random,
and running $Q'$ on $x\circ\pi$ where the string $x\in\Dom(f)$ is
fixed. Compare this to the behavior of $Q'$ on $x\circ\alpha$,
where $\alpha$ is sampled from the hard distribution $\mu$ on $D_{n,r}$.

If $Q'$ did not output $f(x)$ on $x\circ \alpha$
to error at most $1/3$, then we could convert $Q'$ to an
algorithm distinguishing $\pi$ from $\alpha$ with constant error.
This is because $Q'$ outputs $f(x)$ to error at most $7/27<1/3$
on input $x\circ\pi$; hence $Q'$ behaves differently when run on
$x\circ\alpha$ and on $x\circ \pi$. We can convert $Q'$ to
an algorithm $Q''$ which hard codes the input $x$, and receives
either a random $\pi$ from $G$ or a random $\alpha$ from $\mu$
as input. This algorithm $Q''$ will only make $3\Q(f)$ queries
to $\pi$ or $\alpha$, but its acceptance probability
differs by a constant gap between the two distributions, which
(using some standard re-balancing) we can use to distinguish
$G$ from $\mu$ to a constant error.

Now, assuming the distribution
$\mu$ was picked to be hard enough (i.e.\ $\epsilon$ was chosen
sufficiently close to $1/2$), this means that $3\Q(f)$,
the query cost of $Q''$, is at least $\cost_\epsilon(G,r)$.
Since $\cost_\epsilon(G,r)$ is the worst-case quantum query complexity
of a Boolean function, it can be amplified. We conclude that
if $Q'$ failed to output $f(x)$ on input $x\circ\alpha$ (with
$\alpha\leftarrow\mu$) to error at most $1/3$, then we have
$\Q(f)=\Omega(\cost(G,r))$, that is, $\Q(f)>\cost(G,r)/c$
for some universal constant $c$ (from amplification).

Now assume that $\Q(f)\le\cost(G,r)/c$. Then $Q'$ has error at most $1/3$
for computing $f(x)$ when run on $\alpha(x)$, with $\alpha$
chosen from $\mu$. Since $\alpha\in D_{n,r}$
uses at most $r$ alphabet symbols,
a randomized algorithm can simulate $Q'$ simply by picking
$\alpha$ from $\mu$ and querying all the $r$ bits of $x$ used
in the string $x\circ\alpha$, fully determining that string.
This algorithm $R$ uses $r$ queries, and makes at most $1/3$
error, so we conclude that $\R(f)\le r$.

By correctly picking $r$, we conclude that
$\R(f)\le \min\{\,r\in\bN:\cost(G,r)\ge c\Q(f)\,\}$,
as desired. Finally, note that if $\cost(G,r)\ge r^{1/a}/b$,
then by picking $r=(bc\Q(f))^a$ we get
$\cost(G,r)\ge c\Q(f)$. From this it follows that
$\R(f)\le (bc\Q(f))^a=O(\Q(f)^a)$, as desired.
\end{proof}

The upshot of \thm{shuffle} is that we can show a class of group actions
$\mathcal{G}$ does not allow super-polynomial quantum speedups
simply by showing that it is well-shuffling -- that is, by showing
that $G\in\mathcal{G}$ is hard to distinguish from
the set of small-range functions $D_{n,r}$
using a quantum query algorithm.

\section{Showing group actions are well-shuffling}

In this section, we introduce some tools for showing that a collection
of group actions is well-shuffling. Due to \thm{shuffle}, a well-shuffling
collection of group actions does not allow any super-polynomial
quantum speedups for the class of functions symmetric under it, so these tools
can be directly used to show that certain symmetries are not consistent
with large quantum speedups.

\subsection{The symmetric group action}

The first fundamental result is that the class of full symmetric
group actions $S_n$ is well-shuffling.
This was shown by Zhandry \cite{Zha13} in a different context, though
we also provide a simpler proof by a reduction from the collision problem.

\begin{restatable}
{theorem}{zhandry}
\label{thm:zhandry}
There is a universal constant $C$ such that any quantum algorithm
distinguishing a permutation in $S_n$ from a string in $D_{n,r}$
must make at least $r^{1/3}/C$ queries.
\end{restatable}

This theorem says that $S_n$ is hard to distinguish from
$D_{n,r}$ (moreover, Zhandry \cite{Zha13} showed
that the hard distribution over $D_{n,r}$
is uniform, but we do not need this fact).

\begin{proof}
When $n$ is a multiple of $r$, then each $(n/r)$-to-$1$
function has range $r$ and each $1$-to-$1$ function is a permutation;
hence distinguishing $(n/r)$-to-$1$ from $1$-to-$1$ functions
is a sub-problem of distinguishing $D_{n,r}$ from $S_n$.
This sub-problem is the collision problem,
from which an $\Omega(r^{1/3})$
lower bound directly follows \cite{AS04,Amb05,Kut05}.
When $n$ is not a multiple of $r$ but $r\le n/2$, we can just set
$n'=r\lceil n/r\rceil$, and then distinguishing $(n'/r)$-to-$1$
from $1$-to-$1$ functions with domain size $n'$ still
reduces to distinguishing $D_{n,r}$ from $S_n$.
\end{proof}

From \thm{zhandry}, the following two corollaries immediately follow
(in light of \thm{shuffle}).

\begin{corollary}
The set of symmetric group actions $\mathcal{S}=\{S_n\}_{n\in\bN}$
is well-shuffling with power $3$.
\end{corollary}

\begin{corollary}
\label{cor:zhandry}
All (possibly partial) Boolean functions $f$ that are symmetric
under the full symmetric group action $S_n$ satisfy
$\R(f)=O(\Q(f)^3)$.
\end{corollary}

Apart from \thm{zhandry}, the main tools we use to prove
the well-shuffling property are transformations on group actions which
approximately preserve $\cost(G,r)$.
We outline several such transformations and invariances.
Since we prove \thm{zhandry} by a reduction from collision,
and since our main tools from here on out are additional reductions,
it's effectively the case that all lower bounds in this paper work by
reductions from collision.

\subsection{The case of highly-transitive group actions}

We next show that a collection of highly-transitive group actions
is always well-shuffling; we define the notion of highly-transitive
collections below.

\begin{definition}[Highly transitive]\label{def:highly_transitive}
We say a collection $\mathcal{G}$ of group
actions is \emph{highly transitive} if each group action $G\in\mathcal{G}$
is $n^{\Omega(1)}$-transitive, where $n$ is the domain size of $G$.
In other words, $\mathcal{G}$ is highly transitive
if there exist some constants $a,b\in\bN$ such that each $G\in\mathcal{G}$
is $(n^{1/a}/b)$-transitive, where $n$ is the domain size of $G$.
\end{definition}

To show that highly transitive collections of group actions are well-shuffling,
we show that group actions with high transitivity look nearly indistinguishable
from $S_n$ to any quantum algorithm, and that they therefore
share the well-shuffling property of the group actions $S_n$. More formally,
we have the following theorem.

\begin{theorem}[Similar-looking group actions have similar costs]
\label{thm:transitivity_transformation}
Suppose $G$ and $H$ are group actions on $[n]$
and $k\le n$ is a positive integer such that for each
$i_1,i_2,\dots,i_k,j_1,j_2,\dots,j_k\in[n]$, it holds that
\[\left|\Pr_{\pi\leftarrow G}[\forall \ell\;\pi(i_\ell)=j_\ell]
-\Pr_{\pi\leftarrow H}[\forall \ell\;\pi(i_\ell)=j_\ell]\right|
\le n^{-10k}.\]
Then $\cost(H,r)\ge\Omega(\min\{k,\cost(G,r)\})$.
In particular, if $k\ge n^{\Omega(1)}$ and if
$\cost(G,r)\ge r^{\Omega(1)}$, then we have $\cost(H,r)\ge r^{\Omega(1)}$.
\end{theorem}

\begin{proof}
Let $Q$ be a quantum algorithm for distinguishing
$H$ from $D_{n,r}$ which uses $\cost(H,r)$ and achieves worst-case
error $1/3$. If $\cost(H,r)\ge k$, we are done, so assume
$\cost(H,r)<k$. Now, $Q$ can be converted into a polynomial of
degree at most $2\cost(H,r)$ in the variables $z_{ij}$, where
$z_{ij}=1$ if the input $x$ satisfies $x_i=j$ and otherwise $z_{ij}=0$
(see \cite{AS04}).
This polynomial $p$ satisfies $p(x)\in [0,1/3]$ if $x\in D_{n,r}$ and
$p(x)\in[2/3,1]$ if $x\in H$. It has $n^2$ variables and degree
$d=2\cost(H,r)$. We assume it has no monomials that always
evaluate to $0$ (for example, $z_{11}z_{12}$, which is always
$0$ as $x_1$ cannot be both $1$ and $2$), because if it had such
monomials we could just delete them.

We claim that the sum of absolute values of coefficients of $p$ is at most
$n^{3d}$, where $d=2\cost(H,r)$ is its degree. To see this,
first note that there are at most $\binom{n^2}{d}$ monomials of $p$ of degree
$d$; for each such monomial $m$, let $p_m$ be the polynomial consisting
of all terms in $p$ that use a subset of the variables in $m$.
Then the sum of the absolute values of the coefficients of $p$ is at most
$\binom{n^2}{d}$ times the maximum sum of absolute values of the coefficients
in one of the polynomials $p_m$; since $\binom{n^2}{d}\le n^{2d}$, it suffices
to upper bound the sum of absolute values of coefficients of $p_m$
for arbitrary $m$. Now, $m$ consists of $d$ variables $z_{i_tj_t}$ for $t=1,2,\dots d$,
which equal $1$ when $x_{i_t}=j_t$
and equal $0$ otherwise. Consider feeding into the quantum
algorithm an input string where $x_i=*$ when $i\notin\{i_1,i_2,\dots,i_d\}$,
and $x_{i_t}$ is either $j_t$ or $*$ for $t=1,2,\dots d$. The quantum algorithm
will accept the string with some probability between $0$ and $1$, which
means the polynomial $p$ computing the acceptance probability of $Q$ will evaluate to
something between $0$ and $1$. But such inputs ``zero out'' all terms
that use variables outside of $m$, and hence turn $p$ into $p_m$.
From this we can conclude that $p_m$ is bounded in $[0,1]$ for all inputs
it receives in $\B^d$. But polynomials bounded in $[0,1]$ on the Boolean hypercube
can have sum of coefficients at most $5^d$ (one way to analyze would be
to recall that a bounded polynomial in the $\{-1,1\}$ basis has its sum
of squares of coefficients equal to at most $1$, and has at most $2^d$
coefficients, so by Cauchy-Schwartz, the sum of absolute values of coefficients
is at most $2^{d/2}$; converting the $\{-1,1\}$ basis to the $\B$ basis
requires plugging in $(2z-1)$ terms into the variables, which can increase
the sum of absolute values by a factor of at most $3^d$, for a total of at most
$(3\sqrt{2})^d\le 5^d$). Assuming $n\ge 5$, we get an upper bound of $n^{3d}$
on the sum of absolute values of coefficients of $p$.

We have $d\le 2k$, so this sum is also at most $n^{6k}$.
Now, on each input $x$, the expected output of $p(\pi(x))$
when $\pi$ is sampled uniformly from $H$ is a linear combination
of the expectations of the monomials of $p$. For each monomial,
this expectation is just the probability that the monomial
is satisfied, which by the condition on $G$ and $H$
is within $n^{-10k}$ of the expectation under $\pi\leftarrow G$.
It follows that the expectation of $p(\pi(x))$ when $\pi\leftarrow H$
is within $n^{6k}n^{-10k}=n^{-4k}$ of the expectation of $p(\pi(x))$
when $\pi\leftarrow G$. But this expectation is simply the acceptance
probability of $Q$. Hence the acceptance probability of $Q$
on the uniform distribution on $H$ is within $n^{-4k}$
of the acceptance probability of $Q$ on the uniform distribution on $G$.

Since $Q$ distinguishes $H$ from $D_{n,r}$, it distinguishes
the uniform distribution on $H$ from any string in $D_{n,r}$.
Since it does not distinguish the uniform distribution on $H$
from the uniform distribution on $G$, $Q$ must also
distinguish the uniform distribution on $G$ from any input in
$D_{n,r}$ to error $1/3+n^{-4k}$. By amplifying, we can
get this down to error $1/3$, meaning that $\cost(G,r)=O(\cost(H,r))$,
as desired.
\end{proof}

To show that highly-transitive group actions are well-shuffling,
we now only need to show that a $k$-transitive group action $G$
looks like $S_n$ when examining any $k$ bits. This directly follows from
the definition of transitivity.

\begin{corollary}
If $G$ is $k$-transitive, then $\cost(G,r)=\Omega(\min\{k,r^{1/3}\})$,
where the constant in the big-$\Omega$ is universal.
\end{corollary}

\begin{proof}
This follows directly from \thm{transitivity_transformation},
setting $H$ to be the $k$-transitive group action we care about and setting
$G=S_n$. To see this, observe that $k$-transitivity completely determines
$\Pr_{\pi\leftarrow G}[\forall \ell\in[k]\;\pi(i_\ell)=j_\ell]$,
and that both $G$ and $H$ are $k$-transitive; hence this expression
is the same for both $G$ and $H$, and the difference between the two expressions
is exactly $0$ (certainly less than $n^{-10k}$).
\end{proof}

From this, the formal version of \thm{main_transitive} follows.

\begin{corollary}\label{cor:formal_transitive}
If $\mathcal{G}$ is a highly transitive collection of group actions,
then it is well-shuffling, and hence $\R(f)=O(\poly(\Q(f)))$ for
$f\in F(\mathcal{G})$.
\end{corollary}

\subsection{Transformations for graph symmetries}

Next, we introduce some additional transformations on group actions
which approximately preserve the cost; the transformations in this section
will allow us to show that graph property group actions (and several variants of them)
are well-shuffling.

\subsubsection{Transformation for directed graphs}

We start by defining an extension of a group action $G$ on $[n]$ to an
action on $[n]^\ell$. The notation in the definition below
comes from \cite{Ker13}.

\begin{definition}
Let $G$ be a group action on domain $D$, and let $\ell\in\bN$.
Define $G^{(\ell)}$ to be the group action which acts on domain $D^\ell$
by $\pi(i_1,i_2,\dots,i_\ell)=(\pi(i_1),\pi(i_2),\dots,\pi(i_\ell))$
for each $\pi\in G$ (so the number of permutations in $G^{(\ell)}$
is the same as the number of permutations in $G$).

Define $G^{<\ell>}$ to be the group action $G^{(\ell)}$
with domain restricted
to the subset $D^{<\ell>}\subseteq D^\ell$
consisting of all distinct $\ell$-tuples of elements of $D$.
\end{definition}

We show that these transformations both preserve the cost, at least when $\ell$
is constant. We start with $G^{(\ell)}$.

\begin{theorem}
\label{thm:exponentiation}
Let $G$ be a group action on $[n]$, and let $H$ be
the group action $G^{(\ell)}$. Then we have $\cost(H,r^\ell)\ge\cost(G,r)/\ell$.
\end{theorem}

\begin{proof}
Let $Q$ be an algorithm distinguishing $H$ from $D_{n^\ell,r^\ell}$.
Let $\mu$ be the hard distribution for $G$, such that no algorithm
using fewer than $\cost(G,r)$ can distinguish $\mu$ from the
uniform distribution on $G$. Then $\mu$ is a distribution
on $D_{n,r}$. Let $\mu'$ be the distribution on $D_{n^\ell,r^{\ell}}$
that we get by sampling $\alpha\leftarrow\mu$, and returning
$\alpha'$ defined by
$\alpha'(z)=(\alpha(z_1),\alpha(z_2),\dots,\alpha(z_\ell))$
for each $z\in[n]^\ell$ (here we identify $[n^\ell]$ with $[n]^\ell$).
Note that if $\alpha$ has range $r$, then $\alpha'$ has range
at most $r^\ell$.

Then $Q$ distinguishes $\mu'$ from the uniform distribution
on $H$. The latter distribution is the same as what
you get when sampling $\pi$ uniformly from $G$, and returning $\pi'$
defined by $\pi'(z)=(\pi(z_1),\pi(z_2),\dots,\pi(z_\ell))$
for $z$ in the domain of $H$.
This means that $Q$ can be used to distinguish $\mu$ from the
uniform distribution on $G$: all we need is to simulate
every query of $Q$ using $\ell$ queries to the input $\alpha$.
The desired result follows.
\end{proof}

To handle $G^{<\ell>}$, we first observe that restricting the domain
of a group action to some union of its orbits does not decrease its cost.

\begin{lemma}
\label{lem:cutting_into_orbits}
Let $G$ be a group action on $[n]$, and let $S\subseteq[n]$
be a union of orbits of $G$. Let $G'$ be the group action $G$
acting only on $S$. Then $\cost(G',r)\ge\cost(G,r)$.
\end{lemma}

\begin{proof}
We identify $S$ with $[|S|]$ without loss of generality.
If $Q$ distinguishes $G'$ from $D_{|S|,r}$, then
we can turn it into $Q'$ distinguishing $G$ from $D_{n,r}$ by
having $Q'$ run $Q$ and make queries only from $[|S|]$.
\end{proof}

The fact that $G^{<\ell>}$ does not decrease the cost of $G$ too much
then follows as a corollary of \thm{exponentiation} and \lem{cutting_into_orbits}.

\begin{corollary}\label{cor:angle_power}
Let $G$ be a group action on $[n]$, and let $H$ be the group
action $G^{<\ell>}$. Then $\cost(H,r^\ell)\ge\cost(G,r)/\ell$.
\end{corollary}

\begin{proof}
All we need is to note that $G^{<\ell>}$ is the group
action $G^{(\ell)}$ with domain restricted to $[n]^{<\ell>}$,
which is a union of orbits because $\pi\in G^{(\ell)}$ always sends
a tuple with unique entries to another tuple with unique entries
(since a permutation on $[n]$ is applied to each entry).
The desired result then follows from
\thm{exponentiation} and \lem{cutting_into_orbits}.
\end{proof}

We now observe that the transformation $G^{<\ell>}$ immediately allows us to show
that directed graph symmetries are well-shuffling.

\begin{corollary}[Directed graph symmetries]\label{cor:dir_graph}
The set $\mathcal{G}=\{G_k\}_{k\in\bN}$ of all directed graph
symmetries is well-shuffling with power $6$. Here the group action $G_k$
acts on a domain of size $n=k(k-1)$ representing the possible
arcs of a $k$-vertex directed graph, and $G_k$ consists of
all $k!$ permutations on these arcs that act by relabeling the vertices.
\end{corollary}

\begin{proof}
This immediately follows by observing that $G_k=S_k^{<2>}$.
To see this, note that the domain of $G_k$ is the set of all ordered pairs
$(x,y)\in [k]$ with $x\ne y$, which is precisely $[k]^{<2>}$,
and the permutations in $G_k$ are just those in $S_k$ applied to
both coordinates, which is precisely relabeling the vertices.
\cor{angle_power} then gives $\cost(G_k,r)\ge\cost(S_k,\sqrt{r})/2$,
which is at least $\Omega(r^{1/6})$ by Corollary~\ref{thm:zhandry}.
\end{proof}

The collection of directed hypergraph symmetries is similarly well-shuffling.

\begin{corollary}[Directed hypergraph symmetries]
The set $\mathcal{G}_p=\{G_k\}_{k\in\bN}$ consisting of all
$p$-uniform directed hypergraph symmetries is well-shuffling
with power $3p$.
\end{corollary}

\begin{proof}
This follows from the same argument as \cor{dir_graph}.
\end{proof}

\subsubsection{Transformation for undirected graphs}

To handle undirected graphs, we introduce yet another operation on
group actions which approximately preserves the cost.

\begin{theorem}\label{thm:classes}
Let $G$ be a group action acting on $[n]$, and let
$S_1,S_2,\dots,S_k\subseteq[n]$ be a partition of $[n]$
into $k$ equal parts, such that for all $\pi\in G$,
all $t\in[k]$, and all $i,j\in S_t$, the outputs
$\pi(i)$ and $\pi(j)$ lie in the same set $S_{t'}$.
Let $H$ be the group action on $[k]$ induced by $G$,
where for each $\pi\in G$ we have $\pi'\in H$
such that $\pi'(t)=t'$ if $i\in S_t$ and $\pi(i)\in S_{t'}$.
Then $\cost(H,r)\ge\cost(G,r)$.
\end{theorem}

\begin{proof}
Let $Q$ be a quantum algorithm distinguishing $H$ from $D_{k,r}$
using $\cost(H,r)$ queries. We construct
an algorithm $Q'$ for distinguishing $G$ from $D_{n,r}$.
Let $\mu$ be the hard distribution
on $D_{n,r}$ that is hard to distinguish from the uniform
distribution on $G$. The algorithm $Q'$ fixes a unique
$i_t\in S_t$ for each $t=1,2,\dots,k$.
On input $\alpha$ from $G\cup D_{n,r}$,
the algorithm $Q'$ will run $Q$ in the following way:
each query $t\in[k]$ that $Q$ makes will be turned into the query
$i_t\in[n]$ for $\alpha$, and the output $\alpha(i_t)$ will
be converted into the symbol $t'$ such that $\alpha(i_t)\in S_{t'}$
and returned to $Q$. In this way, the algorithm $Q'$ effectively
runs $Q$ on the mapped string $\phi(\alpha)\in[k]^k$, where
$\phi(\alpha)_t$ is the symbol $t'$ such that $\alpha(i_t)\in S_{t'}$.

Now, if $\alpha\in D_{n,r}$, then $\phi(\alpha)\in D_{k,r}$,
while if $\alpha\in G$, we have $\phi(\alpha)\in H$. Since $Q$
distinguishes $H$ from $D_{k,r}$, it follows that $Q'$ distinguishes
$G$ from $D_{n,r}$ using the same number of queries, as desired.
\end{proof}

We are now finally ready to prove the formal version
of \thm{main_graph}, showing that the
collection of (undirected) graph symmetries is well-shuffling.

\begin{definition}[Graph Symmetries]\label{def:graph_sym}
The collection of \emph{graph symmetries} is the set
$\mathcal{G}=\{G_k\}_{k\in \bN}$ of group actions
with $G_k$ acting on $[n]$ with $n=k(k-1)/2$,
such that the domain $[n]$ represents the set
of all possible edges in a $k$-vertex graph,
and $G_k$ acts on these edges and permutes them
in a way that corresponds to relabelling the vertices of the underlying graph.
\end{definition}

\begin{corollary}\label{cor:formal_graph}
The set of all graph symmetries is well-shuffling with power $6$.
Hence $\R(f)=O(\Q(f)^6)$ for functions $f$ symmetric under a graph symmetry.
\end{corollary}

\begin{proof}
Let $G$ be a directed graph symmetry on domain size $k(k-1)$,
and partition this domain into $k(k-1)/2$ sets of size $2$ of the form
$\{(x,y),(y,x)\}$ for $x,y\in[k]$. Then the induced group action $H$
on these sets (from \thm{classes}) is precisely the undirected graph symmetry
on graphs of size $k$. Since $\cost(H,r)\ge\cost(G,r)$, and since
the directed graph symmetries are well-shuffling with power $6$, it follows
that the undirected graph symmetries are also well-shuffling with power $6$.
\end{proof}

Using similar arguments, we can show a similar result for hypergraphs.
\begin{corollary}
For every constant $p \in \bN$, the collection of all $p$-uniform hypergraph symmetries is well-shuffling.
\end{corollary}

\subsubsection{Transformations for bipartite graphs}

We introduce yet more operations on group actions for the case of bipartite graph
symmetries.

\begin{definition}[Product of group actions]
Let $G_1$, and $G_2$ be two group actions acting on $[n_1]$ and $[n_2]$ respectively.
Then the group product action $G_1\times G_2$ is a group action acting on $[n_1 n_2]$ such that
for any $(\pi_1, \pi_2) \in G_1\times G_2$, and any $k\in[n_1]$ and $\ell\in [n_2]$ we have
$(\pi_1, \pi_2)(k, \ell) = (\pi_1(k), \pi_2(\ell))$.
\end{definition}

(In the above, we identify $[n_1]\times[n_2]$ with $[n_1n_2]$.)

\begin{theorem}
\label{thm:group-product}
For all $G_1, G_2$ acting on $[n_1]$ and  $[n_2]$ respectively and for all $r$, $\cost(G_1\times G_2,r^2)\ge \min\{\cost(G_1,r),\cost(G_2,r)\}$.
\end{theorem}

\begin{proof}
Let $H=G_1 \times G_2$ and $m =n_1n_2$.
Let $Q$ be an algorithm distinguishing $H$ from $D_{m,r^2}$.
Let $\mu_1$ be the hard distribution for $G_1$, and let $\mu_2$ be
the hard distribution for $G_2$.
Then $\mu_1$ is a distribution on $D_{n_1,r}$ and $\mu_2$
is a distribution on $D_{n_2,r}$.
Let $\mu'$ be the distribution on $D_{m,r}$ that we get by sampling $\alpha_1 \leftarrow\mu_1$, and $\alpha_2\leftarrow\mu_2$ independently, and returning
$\alpha'=(\alpha(z_1),\alpha(z_2))$.
Note that if $\alpha_1$ and $\alpha_2$ have range $r$, then $\alpha'$ has range
at most $r^2$. Now, since $Q$ distinguishes $D_{m,r^2}$ from $G_1\times G_2$,
it must also distinguish $\mu'$ from the uniform distribution over $G_1\times G_2$,
which itself is the product of the uniform distribution on $G_1$ and the
uniform distribution on $G_2$. Let $\nu_1$ be the uniform distribution on $G_1$,
and let $\nu_2$ be the uniform distribution on $G_2$. Consider
the behavior of $Q$ on $\mu_1\times\nu_2$. It must either distinguish
this distribution from $\mu_1\times\mu_2$, or else from $\nu_1\times\nu_2$
(since it distinguishes $\mu_1\times\mu_2$ and $\nu_1\times\nu_2$ from
each other). In the first case, we can construct $Q'$ which
artificially generates the sample from $\mu_1$ and uses $Q$ to distinguish
$\mu_2$ from $\nu_2$. In the second case, we can construct $Q'$ which
artificially generates the sample from $\nu_2$ and uses $Q$ to distinguish
$\mu_1$ from $\nu_1$.
Hence $\cost(G_1\times G_2,r^2)\ge\min\{\cost(G_1,r),\cost(G_2,r)\}$, as desired.
\end{proof}

\begin{corollary}
The collection $\mathcal{G}$
of all bipartite graph symmetries with equal parts is well-shuffling.
\end{corollary}

\begin{proof}
This immediately follow by observing that bipartite graph symmetries are
the symmetries $S_{k}\times S_{k}$.
Then \thm{group-product} and \thm{zhandry} give the desired result. 
\end{proof}

\subsubsection{Other transformations}

We introduce one final transformation, which merges two group actions into one.
This transformation also does not decrease the cost. While we have no direct
application for it, we will mention this transformation in some discussion
in the next section.

\begin{lemma}[Merger]
\label{thm:extension}
Let $G$ and $H$ be two group actions on $[n]$, and let
$F=\langle G,H\rangle$ be the group action on $[n]$ which is the closure
of $G\cup H$ under composition. Then $\cost(F,r)\ge\cost(G,r)$.
\end{lemma}

\begin{proof}
Since $G$ is a subset of $F$, distinguishing $G$ from $D_{n,r}$
is strictly easier than distinguishing $F$ from $D_{n,r}$.
\end{proof}

\section{Group actions with exponential quantum speedups}
 
In this section, we exhibit some group actions that do allow
super-polynomial quantum speedups. These serve as a barrier
to proving that certain natural classes of group actions are well-shuffling.

To start, note that some
of the most well-known examples of exponential quantum speedups
in query complexity already have some mild symmetries.

\begin{itemize}
\item Period finding (the query task behind Shor's algorithm)
gives a periodic string and asks for the period;
see \cite{Shor94, Cle04} for a full definition. This function
is symmetric under the cyclic group action $Z_n$.

\item Simon's problem \cite{Sim97} promises that the input string
$x$ represents a function with a hidden shift $s\in\B^{\log n}$,
such that $x_i=x_j$ if and only if $i=j$ or $i\oplus j=s$,
and asks to find the hidden shift $s$. It is not hard to convert
this to a decision problem by requiring the function to output
only one bit of information about $s$; in this form,
Simon's problem is symmetric under the group action
which permutes $[n]$ by flipping some bits of the binary representation
of each $i\in [n]$, i.e.\ the group action $Z_2^{\oplus \log n}$.

\item In Forrelation, the input takes the form of two strings $x$ and $y$
of length $n/2$ each, and the task is to estimate the sum
$\sum_{i,j\in[n/2]} (-1)^{\langle i,j\rangle}x_i y_j$;
for a full definition, see \cite{AA15}. Forrelation is symmetric
under the group action which permutes the bits in the binary
representation of each $i$ (the group $S_{\log n}$).

\item Another way to convert
Simon's problem to a decision problem is to define a Boolean function $f$
which outputs $1$ on strings that satisfy the Simon promise
and outputs $0$ on strings that are far from satisfying the promise;
see \cite{BFNR08} for a full definition. This
version of Simon's problem is symmetric under the group action
that can both flip the individual bits of $i$ and permute them,
a group action of order $n\cdot(\log n)!$ which is the merger
of the group actions $S_{\log n}$ and $Z_2^{\log n}$ above.

\end{itemize}

While these examples are not exhaustive, other functions with exponential quantum
speedups tend to have a similar flavor, being symmetric under group actions
which contain only $\poly(n)$ or maybe $n^{O(\log n)}$ permutations
instead of the maximum of $n^{O(n)}$.

This might suggest that in order to get an exponential quantum speedup
we always need mild symmetries, with the order of the group
action being small (compared to the maximum of $n!=n^{\Theta(n)}$).
However, this turns out not to be true.
\thm{large_order_no_speedup_formal} demonstrates that even very large group actions
may still be consistent with exponential quantum speedups.
This means that a characterization of
the group actions which do not allow super-polynomial quantum
advantage must use some richness measure other than the order
of the group action.

\begin{theorem}[Exponential quantum speedup with high symmetry]
\label{thm:large_order_no_speedup_formal}
For infinitely many $n\in\bN$, there is a group action $H_n$ acting on $[n]$
such that $|H_n|=n^{\Omega(n)}$,
and yet there exists a function $f$ that's symmetric under $H_n$
and has $\Q(f)=O(1)$ and $\R(f)=n^{\Omega(1)}$.
\end{theorem}
\begin{proof}
Consider the function $f=\textsc{For}_{\sqrt{n}}\circ\triv_{\sqrt{n}}$ that we obtain by composing the Forrelation function $\textsc{For}$
(as defined in \cite{AA14}) of input size $\sqrt{n}$
with $\sqrt{n}$ copies of the trivial function $\triv$ of size $\sqrt{n}$.
The function $\triv_m$ is a promise problem that only takes two inputs, $0^m$ and
$1^m$. It outputs $0$ on $0^m$ and $1$ on $1^m$.

Observe that $\R(\triv)=\Q(\triv)=1$ (even for computing this function exactly).
From \cite{AA14}, we have
$\Q(\textsc{For}_m)=1$ and $\R(\textsc{For}_m)=\tOmega(\sqrt{m})$.
By composing the quantum algorithm for $\textsc{For}_{\sqrt{n}}$
with the exact quantum algorithm for $\triv$, we get that $\Q(f)=1$.
On the other hand, by taking the hard distribution for
$\R(\textsc{For}_{\sqrt{n}})$ and replacing each $0$ of the input with
$0^{\sqrt{n}}$ and each $1$ with $1^{\sqrt{n}}$,
we clearly get a distribution over inputs to $f$
that is hard for randomized algorithms; it follows that $\R(f)=\tOmega(n^{1/4})$.

On the other hand, we claim that $f$ is highly symmetric. Indeed, each
copy of $\triv_{\sqrt{n}}$ is symmetric under the group action $S_{\sqrt{n}}$,
which has size $\sqrt{n}^{\Theta(\sqrt{n})}=n^{\Theta(\sqrt{n})}$.
Since there are $\sqrt{n}$ copies of $\triv_{\sqrt{n}}$, they are together
symmetric under $S_{\sqrt{n}}\times S_{\sqrt{n}}\times\dots\times S_{\sqrt{n}}=
\left(S_{\sqrt{n}}\right)^{\sqrt{n}}$, a group action of order
$\left(n^{\Theta(\sqrt{n})}\right)^{\sqrt{n}}=n^{\Theta(n)}$.
\end{proof}

This theorem tells us that while we would like to say that sufficiently ``rich''
group actions do not allow exponential quantum speedups, such a richness notion
cannot simply be the order of the group. We note that such a richness notion
also cannot be transitivity: some $1$-transitive group actions allow
exponential quantum speedups and some don't.
For example, as we have seen, graph properties are $1$-transitive
and yet functions symmetric under graph property group actions do not exhibit
exponential quantum speedups. On the other hand, cyclic group actions are
$1$-transitive and they are consistent with exponential quantum speedups
(e.g.\ period finding). Characterizing richness necessary for a group
action to disallow exponential quantum speedups remains a fascinating open
problem.

\section*{Acknowledgements}

We thank Scott Aaronson for many helpful discussions.

\appendix

\section{Proof of the quantum minimax lemma}\label{app:minimax}

We prove \lem{minimax}, which we restate below.

\minimax*

\begin{proof}
By \cite{BSS03}, there is a finite bound $B$ expressible in terms of
$n$ and $|\Sigma|$ on the necessary size of the work space register for a quantum
algorithm solving $f$ to error $\epsilon$. This means the quantum query
algorithms we deal with can be assumed without loss
of generality to have work space size $B$. A quantum algorithm making $T$ queries
can be represented as a sequence of $T$ unitary matrices of size
upper bounded by $B$; this can be arranged as a finite vector of complex numbers.
It is not hard to see that the set of all
such valid quantum algorithms is a compact set.

For a quantum algorithm $Q$, let $\err(Q,x)$ denote the error $Q$ makes
when run on input $x\in\Dom(x)$; this is $\Pr[Q(x)\ne f(x)]$, where
$Q(x)$ is the random variable for the measured output of $Q$ when run on $x$.
We note that $\err(Q,x)$ is a continuous function of $Q$.
Let $v_Q$ be the vector in $\bR^{|\Dom(f)|}$ defined by
$v_Q[x]\coloneqq\err(Q,x)$. Then $v_Q$ is a continuous function of $Q$.
Further, let $V$ be the set of all such vectors $v_Q$ for valid quantum
algorithms $Q$ which make at most $\Q_\epsilon(f)-1$ queries.
Since the set of such valid quantum algorithms is compact and since $v_Q$
is continuous in $Q$, we conclude that $V$ is compact.
Furthermore, we claim that $V$ is convex: this is because for any two
quantum algorithms $Q$ and $Q'$, there is a quantum algorithm $Q''$ which
behaves like their mixture (in terms of its error on each input $x$).

Next, let $\Delta\subseteq\bR^{|\Dom(f)|}$ be the set of all probability
distributions over $\Dom(f)$. Then $\Delta$ is also convex and compact.
Finally, define $\alpha\colon V\times\Delta\to\bR$ by
$\alpha(v,\mu)\coloneqq\bE_{x\leftarrow\mu}v[x]=\sum_{x\in\Dom(f)}\mu[x]v[x]$.
Then $\alpha$ is continuous in each coordinate, and is \emph{saddle}:
that is, $\alpha(\cdot,\mu)$ is convex for each $\mu\in\Delta$
(indeed, it is linear),
and $\alpha(v,\cdot)$ is concave for each $v\in V$ (indeed, it is also linear).
A standard minimax theorem (e.g.\ \cite{Sio58}) then gives us
\[\min_{v\in V}\max_{\mu\in\Delta}\alpha(v,\mu)
=\max_{\mu\in\Delta}\min_{v\in V}\alpha(v,\mu).\]
For the left hand side, it is clear that the maximum over $\mu$ (once the vector $v$
has been chosen) is the same as the maximum over $x\in\Dom(f)$ of $v[x]$.
This makes the left hand side the minimum over $v\in V$ of $\|v\|_{\infty}$,
or equivalently, the minimum worst-case error of quantum algorithms making at most
$\Q_\epsilon(f)-1$ queries. By the definition of $\Q_\epsilon(f)$, this minimum
must be strictly greater than $\epsilon$ (or else $\Q_\epsilon(f)$ would be smaller).
Hence the left hand side is strictly greater than $\epsilon$.

Looking at the right hand side, we see that we get a single distribution
$\mu$ such that every quantum algorithm $Q$ making at most $\Q_\epsilon(f)-1$
queries must make error against $\mu$ which is greater than $\epsilon$, as desired.
\end{proof}

\bibliographystyle{alphaurl}
\phantomsection\addcontentsline{toc}{section}{References} 
\renewcommand{\UrlFont}{\ttfamily\small}
\let\oldpath\path
\renewcommand{\path}[1]{\small\oldpath{#1}}
\bibliography{symmetric}

\newcommand{\etalchar}[1]{$^{#1}$}
\begin{thebibliography}{CCD{\etalchar{+}}03}

\bibitem[AA14]{AA14}
Scott Aaronson and Andris Ambainis.
\newblock The need for structure in quantum speedups.
\newblock {\em Theory of Computing}, 10:133--166, 2014.
\newblock URL: \url{http://theoryofcomputing.org/articles/v010a006/}.

\bibitem[AA15]{AA15}
Scott Aaronson and Andris Ambainis.
\newblock Forrelation: A problem that optimally separates quantum from
  classical computing.
\newblock In {\em Proceedings of the Forty-Seventh Annual ACM on Symposium on
  Theory of Computing}, pages 307--316. ACM, 2015.

\bibitem[AB16]{AB16}
Scott Aaronson and Shalev Ben{-}David.
\newblock Sculpting quantum speedups.
\newblock In {\em 31st Conference on Computational Complexity (CCC)}, pages
  26:1--26:28, 2016.
\newblock \href {http://dx.doi.org/10.4230/LIPIcs.CCC.2016.26}
  {\path{doi:10.4230/LIPIcs.CCC.2016.26}}.

\bibitem[ACL11]{ACL11}
Andris Ambainis, Andrew~M Childs, and Yi-Kai Liu.
\newblock Quantum property testing for bounded-degree graphs.
\newblock In {\em Approximation, Randomization, and Combinatorial Optimization.
  Algorithms and Techniques}, pages 365--376. Springer, 2011.

\bibitem[ACL{\etalchar{+}}19]{ACL+19}
Scott Aaronson, Nai-Hui Chia, Han-Hsuan Lin, Chunhao Wang, and Ruizhe Zhang.
\newblock On the quantum complexity of closest pair and related problems.
\newblock {\em arXiv preprint arXiv:1911.01973}, 2019.

\bibitem[Amb05]{Amb05}
Andris Ambainis.
\newblock Polynomial degree and lower bounds in quantum complexity: Collision
  and element distinctness with small range.
\newblock {\em Theory of Computing}, 1(1):37--46, 2005.

\bibitem[AS04]{AS04}
Scott Aaronson and Yaoyun Shi.
\newblock Quantum lower bounds for the collision and the element distinctness
  problems.
\newblock {\em Journal of the ACM}, 51(4):595--605, July 2004.
\newblock URL: \url{http://doi.acm.org/10.1145/1008731.1008735}, \href
  {http://dx.doi.org/10.1145/1008731.1008735}
  {\path{doi:10.1145/1008731.1008735}}.

\bibitem[BBC{\etalchar{+}}01]{BBC+01}
Robert Beals, Harry Buhrman, Richard Cleve, Michele Mosca, and Ronald De{
  }Wolf.
\newblock Quantum lower bounds by polynomials.
\newblock {\em Journal of the ACM (JACM)}, 48(4):778--797, 2001.
\newblock \href {http://arxiv.org/abs/quant-ph/9802049}
  {\path{arXiv:quant-ph/9802049}}, \href
  {http://dx.doi.org/10.1145/502090.502097} {\path{doi:10.1145/502090.502097}}.

\bibitem[BdW02]{BdW02}
Harry Buhrman and Ronald de~Wolf.
\newblock Complexity measures and decision tree complexity: a survey.
\newblock {\em Theoretical Computer Science}, 288(1):21--43, 2002.
\newblock \href {http://dx.doi.org/10.1016/S0304-3975(01)00144-X}
  {\path{doi:10.1016/S0304-3975(01)00144-X}}.

\bibitem[Ben16]{Ben16}
Shalev Ben{-}David.
\newblock The structure of promises in quantum speedups.
\newblock In {\em 11th Conference on the Theory of Quantum Computation,
  Communication and Cryptography (TQC)}, pages 7:1--7:14, 2016.
\newblock \href {http://dx.doi.org/10.4230/LIPIcs.TQC.2016.7}
  {\path{doi:10.4230/LIPIcs.TQC.2016.7}}.

\bibitem[BFNR08]{BFNR08}
Harry Buhrman, Lance Fortnow, Ilan Newman, and Hein R{\"o}hrig.
\newblock Quantum property testing.
\newblock {\em SIAM Journal on Computing}, 37(5):1387--1400, 2008.

\bibitem[BKT18]{BKT18}
Mark Bun, Robin Kothari, and Justin Thaler.
\newblock The polynomial method strikes back: Tight quantum query bounds via
  dual polynomials.
\newblock In {\em Proceedings of the 50th Annual ACM SIGACT Symposium on Theory
  of Computing}, pages 297--310. ACM, 2018.

\bibitem[B{\v S}13]{BS13}
Aleksandrs Belovs and Robert {\v S}palek.
\newblock Adversary lower bound for the k-sum problem.
\newblock In {\em Proceedings of the 4th Conference on Innovations in
  Theoretical Computer Science}, ITCS '13, pages 323--328, 2013.
\newblock URL: \url{http://doi.acm.org/10.1145/2422436.2422474}, \href
  {http://dx.doi.org/10.1145/2422436.2422474}
  {\path{doi:10.1145/2422436.2422474}}.

\bibitem[BSS03]{BSS03}
Howard Barnum, Michael Saks, and Mario Szegedy.
\newblock Quantum query complexity and semi-definite programming.
\newblock In {\em 18th Conference on Computational Complexity (CCC 2003)},
  pages 179--193, 2003.
\newblock \href {http://dx.doi.org/10.1109/CCC.2003.1214419}
  {\path{doi:10.1109/CCC.2003.1214419}}.

\bibitem[CCD{\etalchar{+}}03]{CCD+03}
Andrew~M Childs, Richard Cleve, Enrico Deotto, Edward Farhi, Sam Gutmann, and
  Daniel~A Spielman.
\newblock Exponential algorithmic speedup by a quantum walk.
\newblock In {\em Proceedings of the thirty-fifth annual ACM symposium on
  Theory of computing}, pages 59--68. ACM, 2003.

\bibitem[Cha18]{Cha18}
Andr{\'e} Chailloux.
\newblock A note on the quantum query complexity of permutation symmetric
  functions.
\newblock {\em 10th Innovations in Theoretical Computer Science Conference
  (ITCS 2019)}, 2018.

\bibitem[Cle04]{Cle04}
Richard Cleve.
\newblock The query complexity of order-finding.
\newblock {\em Information and Computation}, 192(2):162--171, 2004.

\bibitem[Ker13]{Ker13}
Adalbert Kerber.
\newblock {\em Applied finite group actions}, volume~19.
\newblock Springer Science \& Business Media, 2013.

\bibitem[Kut05]{Kut05}
Samuel Kutin.
\newblock Quantum lower bound for the collision problem with small range.
\newblock {\em Theory of Computing}, 1(1):29--36, 2005.

\bibitem[Sho94]{Shor94}
Peter~W Shor.
\newblock Algorithms for quantum computation: Discrete logarithms and
  factoring.
\newblock In {\em Proceedings 35th annual symposium on foundations of computer
  science}, pages 124--134. Ieee, 1994.

\bibitem[Sho97]{Sho97}
Peter~W. Shor.
\newblock Polynomial-time algorithms for prime factorization and discrete
  logarithms on a quantum computer.
\newblock {\em SIAM Journal on Computing}, 26(5):1484--1509, 1997.
\newblock \href {http://arxiv.org/abs/quant-ph/9508027}
  {\path{arXiv:quant-ph/9508027}}.

\bibitem[Sim97]{Sim97}
Daniel~R Simon.
\newblock On the power of quantum computation.
\newblock {\em SIAM journal on computing}, 26(5):1474--1483, 1997.

\bibitem[Sio58]{Sio58}
Maurice Sion.
\newblock On general minimax theorems.
\newblock {\em Pacific Journal of Mathematics}, 8(1):171--176, 1958.
\newblock \href {http://dx.doi.org/10.2140/pjm.1958.8.171}
  {\path{doi:10.2140/pjm.1958.8.171}}.

\bibitem[Ver98]{Ver98}
Nikolai~K. Vereshchagin.
\newblock Randomized boolean decision trees: Several remarks.
\newblock {\em Theoretical Computer Science}, 207(2):329 -- 342, 1998.
\newblock \href {http://dx.doi.org/10.1016/S0304-3975(98)00071-1}
  {\path{doi:10.1016/S0304-3975(98)00071-1}}.

\bibitem[Yao77]{Yao77}
A.~Yao.
\newblock Probabilistic computations: {T}oward a unified measure of complexity.
\newblock {\em Proceedings of the 18th IEEE Symposium on Foundations of
  Computer Science (FOCS)}, pages 222--227, 1977.
\newblock \href {http://dx.doi.org/10.1109/SFCS.1977.24}
  {\path{doi:10.1109/SFCS.1977.24}}.

\bibitem[Zha13]{Zha13}
Mark Zhandry.
\newblock A note on the quantum collision and set equality problems.
\newblock {\em arXiv preprint arXiv:1312.1027}, 2013.
\newblock \href {http://arxiv.org/abs/1312.1027} {\path{arXiv:1312.1027}}.

\end{thebibliography}

\end{document}